\documentclass[runningheads]{tex-includes/llncs}[2018/03/10]

\pdfoutput=1

\input{layout}

\newcommand{\Minion}{\texttt{minion}}

\newcommand{\ie}{i.e.\ }
\newcommand{\eg}{e.g.\ }

\newcommand{\resp}{resp.\ }        %
\newcommand{\vs}{vs.\ }

\newcommand{\pref}[2]{\hyperref[#2]{#1\ref*{#2}}}

\newcommand{\Sec}[2][Sect.]{\pref{#1~}{sec:#2}}
\newcommand{\Fig}[2][Fig.]{\pref{#1~}{fig:#2}}
\newcommand{\Tab}[2][Table]{\pref{#1~}{tbl:#2}}
\newcommand{\Def}[2][Def.]{\pref{#1~}{def:#2}}
\newcommand{\Lem}[2][Lemma]{\pref{#1~}{lem:#2}}
\newcommand{\Thm}[2][Theorem]{\pref{#1~}{thm:#2}}

\newcommand{\Ex}[2][Ex.]{\pref{#1~}{ex:#2}}
\newcommand{\Pty}[1]{(\pref{P}{p:#1})}

\newcommand{\Ev}[1][]{\ifempty{#1}{\alphabet}{\alphabet_{#1}}}  %
\newcommand{\ObsEv}[1][]{\Ev[#1]^\#}           %

\newcommand{\Obs}[1][]{\mathcal{O}_{#1}}      %
\newcommand{\BSys}[1][]{\mathcal{B}_{#1}}     %
\newcommand{\Sys}[1][]{\mathcal{S}_{#1}}      %

\newcommand{\sat}[1][]{\models^s_{#1}}
\newcommand{\nsat}[1][]{\not\models^s_{#1}}
\newcommand{\viol}[1][]{\models^v_{#1}}

\newcommand{\prefix}{\preceq}

\newcommand{\extends}{\succeq}

\renewcommand{\phi}{\varphi}

\renewcommand{\implies}{\Rightarrow}
\renewcommand{\iff}{\Leftrightarrow}
\newcommand{\bigand}{\bigwedge}
\newcommand{\bigor}{\bigvee}

\newcommand{\dmHyp}{\ensuremath{\varphi_{\mathsf{dm}}}}

\newcommand{\dmHypSub}[1]{\ensuremath{\varphi_{#1}}}

\newcommand{\enableLstShortInline}{%
  \lstMakeShortInline[%
    style=base-plain,%
    flexiblecolumns=false,%
    mathescape=true,%
    basicstyle=\tt]@
}
\newcommand{\disableLstShortInline}{\lstDeleteShortInline@}

\newsavebox{\subfloatbox}
\newenvironment{sfloat}[3]{%
  \def\subfloatcap{#1}%
  \def\subfloatlabel{#2}%
  \begin{lrbox}{\subfloatbox}%
  \begin{minipage}[t]{#3}}{%
  \end{minipage}%
  \end{lrbox}%
  \subfloat[\subfloatcap]{\label{\subfloatlabel}\usebox{\subfloatbox}}}

\newcommand{\Into}{\mathrel{\rightarrow}}

\lstdefinestyle{my_style}{
    numberstyle=\tiny,%
    basicstyle=\linespread{0.8}\footnotesize,
    breakatwhitespace=false,
    breaklines=false,
    keepspaces=true,
    numbers=left,
    numbersep=5pt,
    showspaces=false,
    showstringspaces=false,
    showtabs=false,
    tabsize=4,
    frame=single,
    xleftmargin=0.415cm,
    xrightmargin=0.13cm,
    numberbychapter=false
}

\newcommand{\Always}{\LTLsquare}
\newcommand{\Event}{\LTLdiamond} 
\newcommand{\Next}{\LTLcircle}

\newcommand{\U}{\mathbin{\mathcal{U}}}
\newcommand{\F}{\Event}
\newcommand{\G}{\Always}
\newcommand{\X}{\Next}

\newcommand{\DefinedAs}{\,\stackrel{\text{def}}{=}\,}

\newcommand{\DDM}{DDM\xspace}

\DeclareMathOperator{\proj}{proj}

\newcommand{\IN}[1]{\ensuremath{#1_{\textit{in}}}}
\newcommand{\OUT}[1]{\ensuremath{#1_{\textit{out}}}}

\newcommand{\uvar}{\pi}
\newcommand{\vvar}{\pi'}
\newcommand{\pvar}{\tau}
\newcommand{\qvar}{\tau'}

\newcommand{\uin}{\IN{\uvar}}
\newcommand{\vin}{\IN{\vvar}}
\newcommand{\uout}{\OUT{\uvar}}
\newcommand{\vout}{\OUT{\vvar}}
\newcommand{\pin}{\IN{\pvar}}
\newcommand{\qin}{\IN{\qvar}}

\newcommand{\myparagraph}[1]{\paragraph{\textup{\textbf{#1}}}}

\newcommand{\alphabet}{\mathrm{\Sigma}}
\renewcommand{\AP}{\mathsf{AP}}
\newcommand{\V}{\mathcal{V}}

\newcommand{\powerset}{\mathcal{P}}
\newcommand{\ialphabet}{\alphabet^\omega}
\newcommand{\fpowerset}{\powerset_{\textit{fin}}}

\newcommand{\Or}{\mathrel{\vee}}

\newcommand{\Iff}{\mathrel{\leftrightarrow}}

\newcommand{\DefOR}{\ensuremath{\hspace{0.2em}\big|\hspace{0.2em}}}

\newcommand{\Z}{\hphantom{0}}
\newcommand{\PM}[1]{{\tiny$\,\pm$#1}}
\newcommand{\FF}{$\bot$}
\newcommand{\QM}{$?$}
\newcommand{\Eager}{\multicolumn{1}{c}{Eager}}
\newcommand{\Lazy}{\multicolumn{1}{c}{Lazy}}

\begin{document}

\title{Gray-box Monitoring of Hyperproperties}
\subtitle{%
  Extended Version%
  \footnote{This is an extended version of a paper presented at the
    23rd International Symposium on Formal Methods (FM~'19).
    This version contains full proofs, a description of the
    proof-of-concept monitor for DDM, and experimental results that
    were not included in the original publication.
    The original publication is available from Springer at
    \url{https://doi.org/10.1007/978-3-030-30942-8_25}
  }}
\author{
  Sandro Stucki\inst{1} %
  \and
  C\'esar S\'anchez\inst{2} %
  \and \\
  Gerardo Schneider\inst{1} %
  \and
  Borzoo Bonakdarpour\inst{3} %
}
\authorrunning{S.~Stucki et al.}
\institute{%
University of Gothenburg, Sweden, %
\email{sandro.stucki@gu.se,gerardo@cse.gu.se} \and
IMDEA Software Institute, Spain, %
\email{cesar.sanchez@imdea.org} \and
Iowa State University, USA, %
\email{borzoo@iastate.edu}
}
\maketitle              %

\begin{abstract}

  Many important system properties, particularly in security and
  privacy, cannot be verified statically.
  Therefore, runtime verification is an appealing alternative.
  Logics for hyperproperties, such as HyperLTL, support a
  rich set of such properties.
  We first show that {\em black-box} monitoring of HyperLTL is in general
  unfeasible, and suggest a {\em gray-box} approach.
  Gray-box monitoring implies performing analysis of the system at
  run-time, which brings new limitations to monitorabiliy (the
  feasibility of solving the monitoring problem).
  Thus, as another contribution of this paper, we refine the classic
  notions of {\em monitorability}, both for trace properties and
  hyperproperties, taking into account the computability of the
  monitor.
  We then apply our approach to monitor a privacy hyperproperty called
  {\em distributed data minimality}, expressed as a HyperLTL property,
  by using an SMT-based static verifier at runtime.
\end{abstract}

\section{Introduction}
\label{sec:introduction}

Consider a {\em confidentiality} policy $\varphi$ that requires that every
pair of
separate executions of a system agree on the
position of occurrences of some proposition~$a$.
Otherwise, an external observer may learn some sensitive information
about the system. 
We are interested in studying how to build runtime monitors for
properties like $\varphi$, where the monitor receives independent
executions of the system under scrutiny and intend to determine
whether or not the system satisfies the property.
While no such monitor can determine whether the system satisfies
$\varphi$ --- as it cannot determine whether it has observed the whole
(possibly infinite) set of traces --- it may be able to detect
violations.
For example, if the monitor receives finite executions 
$t_1 = \{a\}\{\}\{\}\{a\}\{\}$ and 
$t_2 = \{a\}\{a\}\{\}\{\}\{a\}$, then it is straightforward to see
that the pair $(t_1, t_2)$ violates $\varphi$ (the traces do not agree
on the truth value of $a$ in the second, fourth, and fifth positions).

Now, if we change the policy to $\varphi'$ requiring that, for every
execution, there must exist a different one that agrees with the first
execution on the position of occurrences of $a$, the monitor cannot
even detect violations of $\varphi'$.
Indeed, it is not possible to tell at run-time whether or not for each
execution (from a possibly infinite set), there exists a related one.
Such properties for which no monitor can detect satisfaction or
violation are known as {\em non-monitorable}.

Monitorability was first defined in~\cite{pnueli06psl} as the
problem of deciding whether any extension of an observed
trace would violate or satisfy a property expressed in LTL.
We call this notion \emph{semantic black-box monitorability}.
It is semantic because it defines a decision problem (the existence of
a satisfying or violating trace extension) without requiring a
corresponding decision procedure.
In settings like LTL the problem is decidable and the decision
procedures are well-studied, but in other settings, a property may be
semantically monitorable even though no algorithm to monitor it exists.
This notion of monitorability is ``black-box'' because it only
considers the temporal logic formula to determine the plausibility of
an extended observation that violates or satisfies the formula.
This is the only sound assumption without looking inside the system.
Many variants of this definition followed, mostly for trace
logics~\cite{havelund18runtime} (see also~\cite{bartocci18introduction}).

The definition of semantic monitorability is extended
in~\cite{agrawal16runtime} to the context of {\em
  hyperproperties}~\cite{cs10}.
A hyperproperty is essentially a set of sets of traces, so monitoring
hyperproperties involves reasoning about multiple traces
simultaneously.
The confidentiality example discussed above is a
hyperproperty. 
The notion of monitorability for hyperproperties
in~\cite{agrawal16runtime} also considers whether extensions of an
observed trace, or of other additional observed traces, would violate
or satisfy the property.
An important drawback of these notions of monitorability is that they
completely ignore the role of the system being monitored and the
possible set of executions that it can exhibit to compute a verdict of
a property.

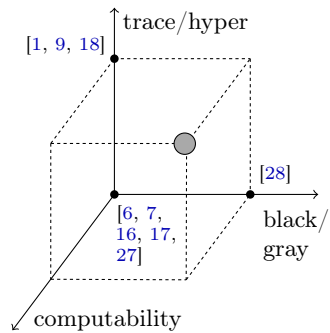
\begin{wrapfigure}{R}{.4\textwidth}
\vspace{-2em}
\centering
  \begin{tikzpicture}[scale=0.4,
  ipe import,
  ipe node/.append style={font=\footnotesize},%
  ipe dash dashed/.style={dash pattern=on 1.2bp off 1.2bp},
  refs/.style={ipe node,font=\scriptsize}]
  \draw[-to]
    (192, 320)
     -- (192, 496);
  \draw[-to]
    (192, 320)
     -- (384, 320);
  \draw[-to]
    (192, 320)
     -- (96, 192);
  \draw[ipe dash dashed]
    (132.05, 368.066)
     -- (192, 448);
  \draw[ipe dash dashed]
    (192, 448)
     -- (320, 448);
  \draw[ipe dash dashed]
    (132.05, 368.066)
     -- (132.018, 240.024);
  \draw[ipe dash dashed]
    (320, 448)
     -- (320, 320);
  \draw[ipe dash dashed]
    (260.003, 368.004)
     -- (320, 448);
  \draw[ipe dash dashed]
    (132.05, 368.066)
     -- (260.003, 368.004);
  \draw[ipe dash dashed]
    (132.018, 240.024)
     -- (259.629, 239.505);
  \draw[ipe dash dashed]
    (260.003, 368.004)
     -- (259.629, 239.505);
  \draw[ipe dash dashed]
    (259.629, 239.505)
     -- (320, 320);
  \node[ipe node]
     at (200, 476) {trace/hyper};
  \node[ipe node, anchor=north west]
     at (332, 306) {\parbox{5ex}{black/\\gray}};
  \node[ipe node]
     at (116, 196) {computability};
  \fill[black]
    (192, 448) circle[radius=4];
  \fill[black]
    (192, 320) circle[radius=4];
  \node[refs, anchor=south east]
     at (184, 452) {\cite{agrawal16runtime,bss18,havelund18runtime}};
  \node[refs, anchor=north west]
      at (194, 310) {\parbox{24pt}{%
          \raggedright%
          \cite{bauer11runtime,bauer07good,falcone12what,%
            finkbeiner17monitoring,pnueli06psl}%
        }};
  \node[refs, anchor=south west]
     at (328, 330) {\cite{zhang12runtime}};
  \fill[black]
    (320, 320) circle[radius=4];
  \filldraw[draw=black, fill=white!66.3!black]
    (258.287, 367.836) circle[radius=10.0678];
\end{tikzpicture}
\caption{The monitorability cube.}
\label{fig:cube}
\vspace{-1.6em}
\end{wrapfigure}

In this paper, we consider a landscape of monitorability aspects along
three dimensions, as depicted in \Fig{cube}.
We explore the ability of the monitor to reason about multiple traces
simultaneously (the trace/hyper dimension).
We first show that a large class of hyperproperties that involve
quantifier alternations are non-monitorable.
That is, no matter the observation, no verdict can ever be declared.
We then propose a solution based on a combination of static analysis
and runtime verification.
If the analysis of the system is completely precise, we call it
\emph{white-box} monitoring.
{\em Black-box} monitoring refers to the classic approach of ignoring
the system and crafting general monitors that provide sound verdicts
for every system.  
In \emph{gray-box} monitoring, the monitor uses an approximate set of
executions, given for example as a model, in addition to the observed
finite execution.
The combination of static analysis and runtime verification allows to
monitor hyperproperties of interest, but it involves reasoning about
possible executions of the system (the black/gray dimension in
\Fig{cube}).
This, in turn, forces us to consider the computability limitations of
the monitors themselves as programs (the computability dimension).

We apply this approach to monitoring a complex hyperproperty of
interest in privacy, namely, {\em data minimization}.
The principle of data minimization (introduced in Article 5 of the EU 
General Data Protection Regulation~\cite{gdpr2012}) from a software 
perspective requires that only data that is semantically used by a program 
should be collected and processed. 
When data is collected from independent sources, the property is
called {\em distributed data minimization} (\DDM) \cite{ASS17dm,PSS18rvh}.
Our approach for monitoring \DDM is as follows.
We focus on detecting violations of \DDM (which we express in HyperLTL
using one quantifier alternation).
We then create a gray-box monitor that collects dynamically potential
witnesses for the existential part.
The monitor then invokes an oracle (combining symbolic execution trees
and SMT solving) to soundly decide the universally quantified inner
sub-formula.
Our approach is sound but approximated, so the monitor may give an
inconclusive answer, depending on the precision of the static
verification.  
\paragraph{Contributions.} In summary, the contributions of this paper
are the following:
\begin{enumerate}[$(1)$]
\item Novel richer definitions of monitorability that consider trace
  and hyper-properties, and the possibility of analyzing the system
  (gray-box monitoring).
  This enables the monitoring, via the combination of static analysis
  and runtime verification, of properties that are non-monitorable in
  a black-box manner.
  Our novel notions of monitorability also cover the computability
  limitations of monitors as programs, which is inevitable once the
  analysis is part of the monitoring process.
\item We express \DDM as a hyperproperty and study its monitorability
  within the richer landscape defined above.
  We then apply the combined approach where the static analysis in
  this case is based on symbolic execution (\Sec{monitoring}).
\item We describe a proof-of-concept implementation of our gray-box
  monitor for \DDM, apply it to some representative examples, and
  present empirical evaluation (\Sec{minion}).
\end{enumerate}
The source code of our implementation is freely available online.%
\footnote{At \url{https://github.com/sstucki/minion/}}

\section{Background}
\label{sec:background}

Let $\AP$ be a finite set of \emph{atomic propositions} and $\alphabet
= 2^{\AP}$ be the finite \emph{alphabet}.
We call each element of $\alphabet$ a \emph{letter} (or an
\emph{event}).
Throughout the paper, $\alphabet^\omega$ denotes the set of all
infinite sequences (called {\em traces}) over $\alphabet$, and
$\alphabet^*$ denotes the set of all finite traces over
$\alphabet$.
For a trace $t \in \alphabet^\omega$ (or $t \in \alphabet^*$), $t[i]$
denotes the $i^{th}$ element of $t$, where $i \in \mathbb{N}$.
We use $|t|$ to denote the length (finite or infinite) of trace $t$.
Also, $t[i, j]$ denotes the subtrace of $t$ from position $i$ up to
and including position $j$ (or $\epsilon$ if $i>j$ or if $i>|t|$).
In this manner $t[0,i]$ denotes the prefix of $t$ up to and including
$i$ and $t[i,..]$ denotes the suffix of $t$ from $i$ (including $i$).

Given a set $X$, we use $\powerset(X)$ for the set of subsets of $X$
and $\fpowerset(X)$ for the set of finite subsets of $X$.
Let $u$ be a finite trace and $t$ a finite or infinite trace.
We denote the concatenation of $u$ and $t$ by $ut$.
Also, $u \prefix t$ denotes the fact that $u$ is a prefix of $t$.
Given a finite set $U$ of finite traces and an arbitrary set $W$ of
finite or infinite traces, we say that $W$ extends $U$ (written
$U \prefix W$) if, for all $u \in U$, there is a $v \in W$, such that
$u \prefix v$.
Note that every trace in $U$ is extended by some trace in $W$ (we call
these \emph{trace extensions}), and that $W$ may also contain
additional traces with no prefix in $U$ (we call these \emph{set
  extensions}).

\subsection{LTL and HyperLTL}

We now briefly introduce LTL and HyperLTL.
The syntax of LTL~\cite{pnueli77temporal} is:
\begin{align*}
\varphi & ::= a  \DefOR \neg \varphi \DefOR \varphi \Or \varphi \DefOR \X \varphi \DefOR
\varphi \U \varphi
\end{align*}
where $a\in \AP$.
The semantics of LTL is given by associating to a formula the set of
traces $t\in \Sigma^\omega$ that it accepts:
\[
\begin{array}{l@{\hspace{2em}}c@{\hspace{2em}}l}
t \models p           & \text{iff} & p\in t[0] \\
t \models \neg\varphi  & \text{iff} & t \not\models \varphi \\
t \models \varphi_1\Or \varphi_2 & \text{iff} & t \models \varphi_1 \text{ or }  t \models \varphi_2 \\
t \models \X \varphi & \text{iff} &  t[1,..] \models \varphi \\
t \models \varphi_1\U \varphi_2 & \text{iff} &  \text{for some $i$, } t[i,..] \models \varphi_2 \text{ and } \text{for all $j <i$, } t[j,..]\models \varphi_1\\
\end{array}
\]
We will also use the usual derived operators
$(\Event\varphi\equiv\textit{true}\,\U\varphi)$ and
$(\Always\varphi\equiv\neg\Event\neg\varphi)$.
All properties expressible in LTL are {\em trace properties} (each
individual trace satisfies the property or not, independently of any other
trace).
Some important properties, such as information-flow security policies
(including confidentiality, integrity, and secrecy), cannot be
expressed as trace properties but require reasoning about two (or
more) independent executions (perhaps from different inputs)
simultaneously.
Such properties are called {\em hyperproperties}~\cite{cs10}.
HyperLTL~\cite{cfkmrs14} is a temporal logic for hyperproperties that
extends LTL by allowing explicit quantification over execution traces.
The syntax of HyperLTL is:
\begin{align*}
\varphi & ::= \forall \pi.\varphi \DefOR \exists \pi.\varphi \DefOR \psi
\hspace{4em}
\psi ::= a_\pi \DefOR \neg\psi \DefOR \psi \Or \psi \DefOR \X \psi \DefOR \psi \U \psi
\end{align*}
A trace assignment $\Pi : \V \rightarrow \alphabet^\omega$ is a
partial function mapping trace variables in $\V$ to infinite traces.
We use $\Pi_\emptyset$ to denote the empty assignment, and $\Pi[\pi
\rightarrow t]$ for the same function as $\Pi$, except that $\pi$ is
mapped to trace $t$.
The semantics of HyperLTL is defined by associating formulas with
pairs $(T,\Pi)$, where $T$ is a set of traces and $\Pi$ is a trace
assignment:
\[
\begin{array}{l@{\hspace{2em}}c@{\hspace{2em}}l}
T, \Pi \models \forall \pi.\varphi & \text{iff} &
\text{for all } t \in T \text{ the following holds } T, \Pi[\pi \rightarrow t] \models
\varphi\phantom{aaaaaaaaaaaaaaa}\\
T, \Pi \models \exists \pi.\varphi & \text{iff} &
\text{there exists } t \in T  \text{ such that } T, \Pi[\pi \rightarrow t]
\models \varphi\\
T, \Pi \models \psi & \text{iff} & \Pi \models \psi
\end{array}
\]
The semantics of the temporal inner formulas is defined in terms of
the traces associated with each path (here $\Pi[i,..]$ denotes the map
that assigns $\pi$ to $t[i,..]$ if $\Pi(\pi)=t$):
\[
\begin{array}{l@{\hspace{1em}}c@{\hspace{2em}}l}
\Pi\models a_{\pi} & \text{iff} & a\in\Pi(\pi)[0]\\
\Pi\models  \psi_1 \Or \psi_2 & \text{iff} & \Pi\models  \psi_1 \text{ or } \Pi\models  \psi_2\\
\Pi\models  \neg \psi  & \text{iff} & \Pi\not\models\psi\\
\Pi\models  \X \psi  & \text{iff} & \Pi[1..]\models  \psi\\
\Pi\models  \psi_1\U \psi_2  & \text{iff} & \text{for some $i$, } \Pi[i,..]\models\psi_2 \text{, and } \text{for all $j<i$ } T,\Pi[j,..]\models\psi_1
\end{array}
\]
We say that a set $T$ of traces satisfies a HyperLTL formula $\varphi$
(denoted $T \models \varphi$) if and only if $T, \Pi_\emptyset \models
\varphi$.

\begin{example} Consider the HyperLTL formula
  $\phi = \forall \pi. \forall \pi'. \G(a_{\pi}\Iff a_{\pi'})$
  and $T = \{t_1, t_2, t_3\}$, where $t_1=\{a,b\}\{a,b\}\{\}\{b\}\cdots$,
  $t_2 = \{a\}\{a\}\{b\}\cdots$ and $t_3 = \{\}\{a\}\{b\}\cdots$
  Although traces $t_1$ and $t_2$ together satisfy $\phi$, $t_3$
  does not agree with the other two, i.e., $a\in
  t_1(0), a\in t_2(0)$, but $a\notin{}t_3(0)$.
Hence, $T\not\models\varphi$.
\end{example}

\subsection{Semantic Monitorability}
\label{sec:monhyp}

{\em Runtime verification} (RV) is concerned with (1) generating a
monitor from a formal specification $\varphi$, and (2) using the
monitor to detect whether or not $\varphi$ holds by observing events
generated by the system at run time.
{\em Monitorability} refers to the possibility of monitoring a property.
Some properties are non-monitorable because no finite observation can
lead to a conclusive verdict.
We now present some abstract definitions to encompass previous notions
of monitorability in a general way.
These definitions are made concrete by instantiating them for example
to traces (for trace properties) or sets of traces (for
hyperproperties), see \Ex{OBinstance} below.
\begin{itemize}
\item {\bf Observation.} We refer to the finite information provided
  dynamically to the monitor up to a given instant as an
  \emph{observation}.
  We use $O$ and $P$ to denote individual observations and $\Obs$ to
  denote the set of all possible observations, equipped with an
  operator $O \prefix P$ that captures the extension of an
  observation.
\item {\bf System behavior.} We use $\BSys$ to denote the universe of
  all possible \emph{behaviors} of a system.
  A behavior $B \in \BSys$ may, in general, be an infinite piece of
  information.
  By abuse of notation, $O \prefix B$ denotes that observation
  $O \in \Obs$ can be extended to a behavior $B$.
\end{itemize}

\begin{example}
  \label{ex:OBinstance}
  When monitoring trace properties such as LTL, we have
  $\Obs=\Sigma^*$, an observation is a finite trace $O\in\Sigma^*$,
  $O\prefix O'$ is the prefix relation on finite strings, and
  $\BSys=\ialphabet$.
  When monitoring hyperproperties such as HyperLTL, an observation is
  a finite set of finite traces $O\subset\Sigma^*$, that is, $\Obs =
  \fpowerset(\alphabet^*)$.
  The relation $\prefix$ is the prefix for finite sets of finite
  traces defined above.
  That is, $O\prefix P$ whenever for all $t\in O$ there is a
  $t'\in P$ such that $t\prefix t'$.
  Finally, $\BSys=\powerset(\ialphabet)$.
\end{example}

\noindent
We say that an observation $O \in \Obs$ \emph{permanently satisfies} a
formula $\varphi$, if every $B \in \BSys$ that extends $O$ satisfies
$\varphi$:
\[
O \sat \varphi \quad \text{iff} \quad \text{for all } B\in\BSys \text{ such that } O\prefix B \text{, } B\models\varphi
\]
where $\models$ denotes the satisfaction relation in the semantics of the
logic.
Similarly, we say that an observation $O \in \Obs$ \emph{permanently
  violates} a formula $\varphi$, if every extension
$B \in \BSys$ violates $\varphi$:
\[
 O \viol \varphi \quad \text{iff} \quad \text{for all } B\in\BSys \text{ such that } O\prefix B \text{, } B\not\models\varphi
\]
Monitoring a system for satisfaction (or violation) of a formula
$\varphi$ is to decide whether a finite observation permanently
satisfies (\resp violates) $\varphi$.
\begin{definition}[Semantic Monitorability]
  \label{def:sem-monitorability}
  A formula~$\varphi$ is {\em (semantically) monitorable} if every
  observation $O$ has an extended observation $P \extends O$,
  such that $P\sat \varphi$ or $P\viol \varphi$.
\end{definition}
A similar definition of monitorability only for satisfaction or only
for violation can be obtained by considering only $P \sat \varphi$
or only $P \viol \varphi$.
Instantiating this definition of monitorability for LTL and finite
traces as observations ($\Obs=\Sigma^*$ and $\BSys=\Sigma^\omega$)
leads to the classic definitions of monitorability for LTL by Pnueli
and Zaks~\cite{pnueli06psl}
(see also~\cite{havelund18runtime}).
Similarly, instantiating the definitions for HyperLTL and observations
as finite sets of finite traces leads to monitorability as introduced
by Agrawal and Bonakdarpour~\cite{agrawal16runtime}.

\begin{example}
  The LTL formula $\G\F a$ is not (semantically) monitorable
  since it requires an infinite-length observation, while formulas
  $\G a$ and $\F a$ are monitorable.
  Similarly, $\forall \pi.\forall \pi.\G (a_\pi \leftrightarrow \neg
  a_{\pi'})$ is monitorable, but $\forall \pi.\exists \pi.\G (a_\pi
  \leftrightarrow \neg a_{\pi'})$ is not, as it requires an
  observation set of infinite size.
  We will prove this claim in detail in
  \Sec{monitorability}.
\end{example}

\section{The Notion of Gray-box Monitoring}
\label{sec:monitorability}

Most of the previous definitions of monitorability make certain
assumptions:
\begin{inparaenum}[(1)]
\item the logics are trace logics, \ie do not cover hyperproperties,
\item the system under analysis is black-box in the sense that every
  further observation is possible,
\item the logics are tractable, in that the decision problems of
  satisfiability, liveness, etc.\ are decidable.
\end{inparaenum}
We present here a more general notion of monitorability by challenging
these assumptions.

\subsection{The Limitations of Monitoring Hyperproperties}
\label{sec:limitations}

\newcommand{\PPred}{F}

Earlier work on monitoring hyperproperties is restricted to the
quantifier alter\-nation-free fragment, that is either $\forall^*.\psi$
or $\exists^*.\psi$ properties.
We establish now an impossibility result about the monitorability of
formulas of the form $\forall \pi.\exists \pi'. \Always \PPred$, where
$\PPred$ is a state predicate.
That is, $\PPred$ is formed by atomic propositions, $a_\pi$ or $a_{\pi'}$
and Boolean combinations thereof, and can be evaluated given two
valuations of the propositions from $\AP$, one from each path $\pi$
and $\pi'$ at the current position.
For example, the predicate $\PPred = (a_\pi\leftrightarrow \neg
a_{\pi'})$ for $\AP=\{a\}$ depends on the valuation of $a$ at the
first state of paths $\pi$ and $\pi'$.
We use $v$ and $v'$ in $\PPred(v,v')$ to denote that $\PPred$ uses two
copies of the variables $v$ (one copy from $\pi$ and another from
$\pi'$).
A predicate $\PPred$ is \emph{reflexive} if for all valuations $v\in
2^\AP$, $\PPred(v,v)$ is true.
A predicate $\PPred$ is \emph{serial} if, for all $v$, there is a $v'$ such
that $\PPred(v,v')$ is true.

\newcounter{thm-monitorability}
\setcounter{thm-monitorability}{\value{theorem}}

\begin{theorem}
  \label{thm:monitorability}
  A HyperLTL formula of the form $\psi = \forall \pi.\exists \pi'.\Always
  \PPred$ is non-monitorable if and only if $\PPred$ is non-reflexive and
  serial.
\end{theorem}

\begin{proof}
  Let $\varphi$ be $\forall \pi\exists \pi'.\Always \PPred$.
  We first observe that if $\PPred$ is serial, then the universal set
  $\Sigma^\omega$ is a model of $\varphi$, i.e.
  $\Sigma^\omega\models\varphi$.
  We show the two directions separately.
  \begin{itemize}
  \item ``$\Leftarrow$''.
    Assume that $\PPred$ is non-reflexive and serial, and let $U$ be an
    arbitrary observation.
    We show an infinite extension of $U$ that violates $\varphi$ and
    another infinite extension of $U$ that satisfies $\varphi$,
    concluding that no observation has a finite extension that
    permanently satisfies or violates $\varphi$, that is, $\varphi$ is
    not monitorable.
    As mentioned above, since $\PPred$ is serial, $\Sigma^\omega$ is a
    model of $\varphi$ and $\Sigma^\omega$ extends $U$.
    Now, assume that all traces in $U$ have the same length
    (otherwise, extend the shorter traces arbitrarily).
    Then, pick $v$ such that $\PPred(v,v)$ is false (recall that $\PPred$ is
    non-reflexive so such a $v$ must exist), and consider the set of
    infinite observations
    $V = \{ uvt \;|\; u\in U, t\in\Sigma^\omega\}$.
    Since $v$ appears at the same position in all strings in $B$, it
    follows that $B\not\models\varphi$.
    \item ``$\Rightarrow$''.
    If $\PPred$ is reflexive then $\varphi$ holds for every non-empty set
    of infinite words by picking the same trace for $\pi$ and $\pi'$.
    Therefore $\varphi$ is monitorable (in fact, guaranteed to be
    permanently satisfied for any observation).
    Otherwise, assume that $\PPred$ is not serial, so for some $v$ and for all
    $v'$, $\PPred(v,v')$ is false.
    Consider an arbitrary observation $U$ and extend one $u\in U$ into
    $uv$.
    The observation obtained permanently violates $\varphi$ because
    taking $\pi$ to be $uv$ cannot be matched at the position where
    $v$ occurs by any trace for $\pi'$.
  \end{itemize}
  This finishes the proof.
  \qed
\end{proof}

The fragment of $\forall\exists$ properties captured by
\Thm{monitorability} is very general (and this result can be easily
generalized to $\forall^+\exists^+$ hyperproperties).
First, the temporal operator is just safety (the result can be
generalized for richer temporal formulas).
Also, every binary predicate can be turned into a non-reflexive
predicate by distinguishing the traces being related.
Moroever, many relational properties, such as non-interference and
DDM, contain a tacit assumption that only distinct traces are being
related.
Seriality simply establishes that $\PPred$ cannot be falsified by only
observing the local valuation of one of the traces.
Intuitively, a predicate that is not serial can be falsified by
looking only at one of the traces, so the property is not a proper
hyperproperty.
The practical consequence of \Thm{monitorability} is that many
hyperproperties involving one quantifier alternation cannot be
monitored.

\subsection{Gray-box Monitoring. Sound and Perfect Monitors}

To overcome the negative non-monitorability result, we exploit
knowledge about the set of traces that the system can produce
(gray-box or white-box monitoring).
Given a system that can produce the set of system behaviors $\Sys
\subseteq \BSys$, we parametrize the notions of permanent satisfaction
and permanent violation to consider only  behaviors in $\Sys$:
\begin{align*}
O \sat[\Sys] \varphi & \quad \text{ iff } \quad
\text{ for all $B \in \Sys$ such that $O\prefix B$, $B\models \varphi$}\\
O \viol[\Sys] \varphi & \quad \text{ iff } \quad
\text{ for all $B \in \Sys$ such that $O\prefix B$, $B\not\models \varphi$}
\end{align*}
First, we extend the definition of monitorability
(\Def{sem-monitorability} above) to consider the system under
observation.

\begin{definition}[Semantic Gray-Box Monitorability]
  \label{def:gb-monitorability}
  A formula~$\varphi$ is {\em semantically gray-box monitorable} for a
  system $\Sys$ if every observation $O$ has an extended observation
  $P \extends O$ in $\Sys$, such that $P\sat[\Sys] \varphi$ or
  $P\viol[\Sys] \varphi$.
\end{definition}

\noindent In this definition, monitors must now analyze and decide
properties of extended observations which is computationally not
possible with full precision for sufficiently rich system
descriptions.

We now introduce a novel notion of monitors that consider $\Sys$ and
the computational power of monitors (the diagonal dimension in
\Fig{cube}).
A \emph{monitor} for a property $\phi$ and a set of traces $\Sys$ is a
\emph{computable} function $M_{\Sys} \colon \Obs \to \set{\top,\bot,?}$ that,
given a finite observation $O$, decides a \emph{verdict} for $\phi$:
$\top$ indicates success, $\bot$ indicates failure, and $?$ indicates
that the monitor cannot declare a definite verdict given only $u$.
To avoid clutter, we write $M$ instead of $M_{\Sys}$ when the system
is clear from the context.
The following definition captures when a monitor for a property
$\varphi$ can give a definite answer.
\begin{definition}[Sound monitor]
  Given a property $\varphi$ and a set of behaviors $\Sys$, a monitor
  $M$ is {\em sound} whenever, for every observation $O \in \Obs$,
  \begin{enumerate}
  \item if $O \sat[\Sys] \phi$, then $M(O)=\top$ or $M(O)={}?$,
  \item if $O \viol[\Sys] \phi$, then $M(O)=\bot$ or $M(O)={}?$,
  \item otherwise $M(O)={}?$.
  \end{enumerate}
  \label{def:sound-monitor}
\end{definition}
If a monitor is not sound then it is possible that an extension of $O$
forces $M$ to change a $\top$ to a $\bot$ verdict, or vice-versa.
The function that always outputs $?$ is a sound monitor for any
property, but this is the least informative monitor.
A \emph{perfect monitor} precisely outputs whether satisfaction or
violation is inevitable, which is the most informative monitor.

\begin{definition}[Perfect Monitor]
  Given a property $\varphi$ and a set of traces $\Sys$, a monitor $M$ is
  {\em perfect} whenever, for every observation $O\in \Obs$,
  \begin{enumerate}
  \item if $O \sat[\Sys] \phi$ then $M(O)=\top$,
  \item if $O \viol[\Sys] \phi$ then $M(O)=\bot$,
  \item otherwise $M(O)={}?$.
  \end{enumerate}
  \label{def:perfect-monitor}
\end{definition}
Obviously, a perfect monitor is sound.
Similar definitions of perfect monitor only for satisfaction (\resp
violation) can be given by forcing the precise outcome only for
satisfaction (\resp violation).

A black-box monitor is one where every behavior is potentially
possible, that is $\Sys=\BSys$.
If the monitor uses information about the actual system, then we say
it is {\em gray-box} (and we use \emph{white-box} when the monitor can
reason with absolute precision about the set of traces of the system).
In some cases, for example to decide instantiations of a $\forall$
quantifier, a satisfaction verdict that is taken from $\Sys$ can be
concluded for all over-approximations (dually under-approximations for
violation and for $\exists$).
For space limitations, we do not give the formal details here.

Using \Def[Defs.]{sound-monitor} and\Def[]{perfect-monitor}, we
can add the computability aspect to capture a stronger definition of
monitorability.
Abusing notation, we use $O\in \Sys$ to say that the observation $O$
can be extended to a trace allowed by the system.

\begin{definition}[Strong Monitorability]
  \label{def:strong-monitorability}
  A property $\varphi$ is {\em strongly monitorable} for a system
  $\Sys$ if there is a sound monitor $M$ s.t.\ for all
  observations $O\in\Obs$, there is an extended observation
  $P\in \Sys$ for which either $M(P)=\top$ or $M(P)=\bot$.
\end{definition}
A property is strongly monitorable for satisfaction if the extension
with $M(P)=\top$ always exists (and analogously for violation).
In what follows we will use the term {\em monitorability} to refer to
strong monitorability whenever no confusion may arise.
It is easy to see that if a property is not semantically monitorable,
then it is not strongly monitorable, but in rich domains, some
semantically monitorable properties may not be strongly monitorable.
One trivial example is termination for deterministic programs (that
is, the halting problem).
Given a prefix of the execution of a deterministic program, either the
program halts or it does not, so termination is monitorable in the
semantics sense.
However, it is not possible to build a monitor that decides the
halting problem.
\begin{lemma}
  If $\varphi$ is strongly monitorable, then $\varphi$ is semantically
  monitorable.
\end{lemma}
A property may not be monitorable in a black-box manner, but
monitorable in a gray-box manner.
In the realm of monitoring of LTL properties, strong and semantic
monitorability coincide for finite state systems
(see~\cite{zhang12runtime}) both black-box and gray-box (for finite
state systems), because model-checking and the problem of deciding
whether a state of a B\"uchi automaton is live are decidable.

Following~\cite{bss18} we propose to use a combination of static
analysis and runtime verification to monitor violations of
$\forall^+\exists^+$ properties (or dually, satisfactions of
$\exists^+\forall^+$).
The main idea is to collect candidates for the outer
$\exists$ part dynamically and use static analysis at runtime to
over-approximate the inner $\forall$ quantifiers.

\section{Monitoring Distributed Data Minimality}
\label{sec:monitoring}

In this section, we describe how to monitor DDM, which can be expressed
as a hyperproperty of the form $\forall^+\exists^+$.
The negative non-monitotabiliy result from \Sec{limitations} can be
generalized to $\forall^+\exists^+$ hyperproperties.
In the particular case of DDM, although we mainly deal with the
input/output relation of functions and are not concerned with infinite
temporal behavior, we still need to handle possibly infinite set
extensions $\Sys$ for black-box monitoring.
In the remainder of this section, we discuss the following, seemingly
contradictory aspects of DDM:
\begin{enumerate}[(P1)]
\item DDM is not semantically \emph{black-box} monitorable,
  \label{p:bb-nonmon}
\item DDM is semantically \emph{white-box} monitorable (for programs
  that are not DDM),
  \label{p:wb-mon}
\item checking DDM statically is undecidable,
  \label{p:undec}
\item DDM is strongly gray-box monitorable for violation, and we give
  a \emph{sound monitor}.
  \label{p:smon}
\end{enumerate}
The apparent contradictions are resolved by careful analysis of DDM
along the different dimensions of the monitorability cube
(\Fig{cube}).

We will show how to monitor DDM and similar hyperproperties using a
gray-box approach.
In our approach, a monitor can decide at run time the existence of
traces using a limited form of static analysis.
The static analyzer receives the finite observation $O$ collected by
the monitor, but not the future system behavior.
Instead it must reason under the assumption that any system behavior
in $\Sys$ that is compatible with $O$, may eventually occur.
For example, given an $\exists\forall$ formula, the outer existential
quantifier is instantiated with a concrete set $U$ of runtime traces,
while possible extensions of $U$ provided by static analysis can be
used to instantiate the inner universal quantifier.

\subsection{DDM Preliminaries}
\label{sec:ddm-prelim}

We briefly recapitulate the formal notion of data-minimality
from~\cite{ASS17dm}.
Given a function $f \colon I \to O$, the problem of data minimization
consists in finding a \emph{preprocessor} function $p\colon I \to I$,
such that $f=f\comp p$ and $p=p\comp p$.
The goal of $p$ is to limit the information available to $f$ while
preserving the behavior of $f$.

There are many possible such preprocessors (\eg the identity
function), which can be ordered according to the information they
disclose, that is, according to the
subset relation on their \emph{kernels}.
The kernel $\ker(p)$ of a function $p$ is defined as the equivalence
relation $(x, y) \in \ker(p) \text{ iff } p(x) = p(y)$.
The smaller $\ker(p)$ is, the more information $p$ discloses.
The identity function is the worst preprocessor since it discloses all
information (its kernel is equality --- the least equivalence
relation).
An optimal preprocessor, or \emph{minimizer}, is one that discloses
the least amount of information.

A function $f$ is \emph{monolithic data-minimal} (MDM), if it fulfills
either of the following equivalent conditions:
\begin{enumerate}
\item the identity function is a minimizer for $f$,
\item $f$ is injective.
\end{enumerate}
Condition~1.\ is an information-flow-based characterization
that can be generalized to more complicated settings in a straightforward
fashion.
Condition~2.\ is a purely logical or \emph{data-based} characterization
more suitable for implementation in e.g.\ a monitor.

MDM is the strongest form of data minimality, where one assumes that
all input data is provided by a single source and thus a single
preprocessor can be used to minimize the function.
If inputs are provided by multiple sources (called a distributed
setting) and access to the system implementing $f$ is restricted, it
might be impossible to use a single preprocessor.
For example, consider a web-based auction system that accepts bids
from $n$ bidders, represented by distinct input domains
$I_1, \dotsc, I_n$, and where concrete bids $x_k \in I_k$ are
submitted remotely.
The auction system must compute the function
\( m(x_1, \dotsc, x_n) = \max_k \set{ x_k } \),
which is clearly non-injective and, hence, non-MDM.
In this case, a single, monolithic minimizer cannot be used since
different bidders need not have any knowledge of each other's bids.
Instead, bidders must try to minimize the information contained in
their bid locally, in a distributed way, before submitting it to the
auction.

The problem of \emph{distributed data minimization} consists in
building a collection $p_1, \dotsc, p_n$ of $n$ independent
preprocessors $p_k \colon I_k \to I_k$ for a given function
$f \colon I_1 \times \dotsm \times I_n \to O$,
such that their parallel composition
$p(x_1, \dotsc, x_n) = (p_1(x_1), \dotsc, p(x_n))$ is a preprocessor
for $f$.
Such composite preprocessors are called \emph{distributed}, and
a distributed preprocessor for $f$ that discloses the least amount of
information is called a \emph{distributed minimizer} for $f$.
Then, one can generalize the (information-flow) notion of
data-minimality to the distributed setting as follows.
The function $f$ is \emph{distributed data-minimal} (DDM) if the
identity function is a distributed minimizer for $f$.
Returning to our example, the maximum function $m$ defined above is
DDM.
As for MDM, there is an equivalent, data-based characterization of
DDM defined next.
\begin{definition}[distributed data minimality~\cite{ASS17dm,PASS18corr}]
  \label{def:datamin}
  A function $f$ is \emph{distributed data-minimal (DDM)} if, for all
  input positions $k$ and
  all $x, y \in I_k$ such that $x \neq y$, there is some
  $z \in I$, such that $f(z[k \mapsto x]) \neq f(z[k \mapsto y])$.
\end{definition}
We use \Def{datamin} to explore how to monitor DDM.
In the following, we assume that the function $f\colon I_1 \times
\dotsm \times I_n \to O$ has at least two arguments ($n \geq 2$).
Note that for unary functions, DDM coincides with MDM.
Since MDM is a $\forall^+$-property (involving no quantifier
alternations), most of the challenges to monitorability discussed here
do not apply~\cite{PSS18rvh}.
We also assume, without loss of generality, that the function $f$
being monitored has only nontrivial input domains, \ie
$\abs{I_k} \geq 2$ for all $k = 1, \dotsc n$.
If $I_k$ is trivial then this constant input can be ignored.
Finally, note that checking DDM statically is undecidable
\Pty{undec} for sufficiently rich programming languages
\cite{ASS17dm}.

\subsection{DDM as a Hyperproperty}
\label{sec:dmhyp}

We consider data-minimality for total functions $f \colon I \to O$.
Our alphabet, or set of events, is the set of possible
\emph{input-output (I/O) pairs} of $f$, \ie
$\Ev[f] = I \times O$.
Since a single I/O pair $u = (\IN{u}, \OUT{u}) \in \Ev[f]$ captures an
entire run of $f$, we restrict ourselves to observing singleton
traces, \ie traces of length $\abs{u} = 1$.
In other words, we ignore any temporal aspects associated with the
computation of $f$.
This allows us to use first-order predicate logic --- without any
temporal modalities --- as our specification logic.

DDM is a hyperproperty, expressed as a predicate over sets of traces,
even though the traces are I/O pairs.
The set of observable behaviors $\Obs[f]$ of a given $f$ consists of
all \emph{finite sets} of I/O pairs $\Obs[f] = \fpowerset(\Ev[f])$.
The set of all possible system behaviors
$\BSys[f] = \powerset(\Ev[f])$ additionally includes \emph{infinite
  sets} of I/O pairs.

\begin{example}\label{ex:addition}
  Let $f\colon \NN \times \NN \to \NN$ be the addition function on
  natural numbers, $f(x, y) = x + y$.
  Then $I = \NN \times \NN$, $O = \NN$, and a valid trace
  $u \in \Ev[f]$ takes the form $u = ((x, y), z)$,	 where $x$, $y$ and
  $z$ are all naturals.
  Both $U = \set{((1, 2), 3), ((2, 1), 3)}$ and
  $V = \set{((1, 1), 3)}$ are considered observable behaviors
  $U, V \in \Obs[f]$, even though $V$ does not correspond to a valid
  system behavior since $f(1, 1) \neq 3$.
  Remember that we do not discriminate between valid and invalid
  system behaviors in a black-box setting.
\end{example}

We now express DDM as a hyperproperty, using HyperLTL, but with only
state predicates (no temporal operators).
Given a tuple $x = (x_1, x_2, \dotsc, x_n)$, we write $\proj_i(x)$ or
simply $x_i$ for its $i$-th projection.
Given an I/O pair $u = (x, y)$ we use $\IN{u}$ for the input component
and $\OUT{u}$ for the output component (that is $\IN{u}=x$ and
$\OUT{u}=y$).
Given trace variables $\uvar, \vvar$, we define
\begin{align*}
  \operatorname{output}(\uvar, \vvar) &\DefinedAs \uout = \vout
  && \text{$\uvar$ and $\vvar$ agree on their output,}\\
  \operatorname{same}_i(\uvar, \vvar) &\DefinedAs \proj_i(\uin) = \proj_i(\vin)
  && \text{$\uvar$ and $\vvar$ agree on the $i$-th input,}\\
  \operatorname{almost}_i(\uvar, \vvar) &
  \DefinedAs \bigand_{k \neq i} \proj_k(\uin) = \proj_k(\vin)
  && \text{\parbox{.4\textwidth}{%
     $\uvar$ and $\vvar$ agree on all but the\\
     \phantom{$\uvar$ and $\vvar$} $i$-th input}}
\end{align*}

\begin{example}
  Let $u = ((1, 2), 3)$, $u' = ((2, 1), 3)$, and
  $\Pi = \set {\uvar \mapsto u, \vvar \mapsto u'}$.
  Then $\Pi \models \operatorname{output}(\uvar, \vvar)$, but
  $\Pi \not\models \operatorname{same}_1(\uvar, \vvar)$ and
  $\Pi \not\models \operatorname{almost}_1(\uvar, \vvar)$.
\end{example}

We define DDM for input argument $i$ as follows:
\begin{align*}
  \dmHypSub{i} &\; = \;
  \forall \uvar.\forall \vvar.\exists \pvar.\exists \qvar.\;
  \neg\operatorname{same}_i(\uvar, \vvar) \Into
  \left(\begin{aligned}
      \operatorname{same}_i(\uvar, \pvar) &\land
      \operatorname{same}_i(\vvar, \qvar) \; \land \\
      \operatorname{almost}_i(\pvar, \qvar) &\land
      \neg\operatorname{output}(\pvar, \qvar)
    \end{aligned}\right)
\end{align*}
In words: given any pair of traces $\uvar$ and $\vvar$, if $\uin$ and
$\vin$
differ in their $i$-th position, then there must be some
common values $z$ for the remaining inputs, such that the outputs of
$f$ for $\pin = z[i \mapsto \proj_i(\uin)]$ and $\qin = z[i \mapsto
\proj_i(\vin)]$ differ.
Note that $z$ does not appear in $\dmHypSub{i}$ directly, instead it is
determined implicitly by the (existentially quantified) traces $\pvar$
and $\qvar$.
Finally, \emph{distributed data minimality} for $f$ is defined as
$$\dmHyp \; = \; \bigand_{i = 1}^n \varphi_i.$$

The property $\dmHyp$ follows the same structure as the logical
characterization of DDM from \Sec{ddm-prelim}.
The universally quantified variables range over the possible inputs at
position $i$, while the existentially quantified variables $\pvar$ and $\qvar$
range over the other inputs and the outputs.
Note also that, given the input coordinates of $\uvar$, $\vvar$, and $\pvar$, all
the output coordinates, as well as the input coordinates of $\qvar$, are
uniquely determined.\footnote{For simplicity, even though $\dmHyp$ is
  not in prenex normal form, it is a finite conjunction of
  $\forall\forall\exists\exists$ formulas in prenex normal form so a
  finite number of monitors can be built and executed in parallel, one
  per input argument.}

\begin{example}
  Consider again $U = \set{((1, 2), 3), ((2, 1), 3)}$ and
  $V = \set{((1, 1), 3)}$ from \Ex{addition}.
  Then, $V \models \dmHyp$ trivially holds, but $U \not\models \dmHyp$
  because when $\Pi(\uvar) \neq \Pi(\vvar)$ there is no choice of
  $\Pi(\pvar), \Pi(\qvar) \in U$ for which
  $\Pi \models \neg\operatorname{output}(\pvar, \qvar)$ holds.
\end{example}
Note that, in the above example, $V \models \dmHyp$ holds despite the
fact that $V$ is not a valid behavior of the example
function~$f(x, y) = x + y$.
Indeed, whether or not $U \models \dmHyp$ holds for a given $U$ is
independent of the choice of $f$.
In particular,
$\Ev[f] \models \dmHyp$, for any choice of $f$ regardless of whether
$f$ is data-minimal or not.
This is already a hint that the notion of semantic black-box
monitorability is too weak to be useful when monitoring $\dmHyp$.
Since $\Ev[f]$ is a model of $\dmHyp$, no observation $U$ can have an
extension that permanently violates $\dmHyp$.
As we will see shortly, gray-box monitoring does not suffer from this
limitation.
Monitorability of DDM for violations becomes possible once we exclude
potential models such as $\Ev[f]$ which do not correspond to valid
system behaviors.

\paragraph{Remark.}
Note that though our definition and approach work for general
(reactive) systems, the DDM example is admittedly a non-reactive
system with traces of length 1.
This, however, is not a limitation of the approach.
Extending DDM for reactive systems is left as future work.

\subsection{Properties of DDM}

Since $\dmHyp$ is a $\forall^+\exists^+$ property, it should not come
as a surprise that it is \emph{not} semantically black-box monitorable
in general \Pty{bb-nonmon}.

\begin{lemma}[black-box non-monitorability]\label{lem:bb-nonmon}
  Assume $f\colon I \to O$, then $\dmHyp$ is semantically black-box
  monitorable iff $I$ is finite.
\end{lemma}
\begin{proof}
  We first treat the case where $I$ is finite.
  Assume $I$ and $O$ each contain at least two elements.
  Smaller I/O domains correspond to degenerate cases for which
  semantic black-box monitorability is easy to show, so we omit them
  here.

  Let $U \subseteq \Obs$ be a finite set of traces.
  We need to show that there is a finite extension $V \extends U$ that
  permanently satisfies or violates $\dmHyp$.
  Pick $V = \Ev[f] = I \times O$.
  Clearly, this is the largest observation in $\Obs$, so any property
  satisfied by $V$ is also permanently satisfied by $V$.
  Hence it suffices to show that $V \sat \dmHyp$.

  Let $u$, $u'$, $w$ be arbitrary I/O pairs, $o \neq o' \in O$ a pair
  of distinct outputs, and $i$ an arbitrary input position.
  Define $v = (\IN{w}[i \mapsto \proj_i(\IN{u})], o)$ and
  $v' = (\IN{w}[i \mapsto \proj_i(\IN{u'}), o')$.
  Then $u$, $u'$ and $v$, $v'$ are all in $V$, and it is easy to check
  that $\dmHypSub{i}$ holds if the quantified variables are
  instantiated to these traces in the given order.
  In other words $V \sat \dmHypSub{i}$ for all $i$, and hence $V$
  permanently satisfies $\dmHyp$.

  Conversely, assume that $I$ is infinite, and let $U$ again be a
  finite set of traces.
  To show that $U$ neither permanently satisfies nor permanently
  violates $\dmHyp$, it is sufficient to exhibit a pair of extensions
  $T_s, T_v \extends U$ that satisfy and violate $\dmHyp$,
  respectively.
  For $T_s$, we pick $T_s = \Ev[f] = I \times O$.
  By the same argument as given above (for the finite case), we have
  $T_s \sat \dmHyp$.

  We have to work slightly harder to construct $T_v$.
  Since $I$ is infinite but $U$ is finite, there must be an input
  position $i$ and a pair of distinct elements $x \neq x' \in I_i$
  such that no trace in $U$ has $x$ or $x'$ as its $i$-th input.
  Pick some arbitrary trace $w \in \Ev[f]$, and let
  $v = w[i \mapsto x]$ and $v' = w[i \mapsto x']$.
  By construction, $v, v' \notin U$, so $T_v = U \cup \set{ v, v' }$
  is a strict extension of $U$.
  To show that $T_v$ does indeed violate $\dmHyp$, it is sufficient to
  show that $T_v \viol \dmHypSub{i}$.
  Pick $v, v'$ to instantiate $\uvar$ and $\vvar$.
  Then $\proj_i(\IN{w}) = x \neq x' = \proj_i(\IN{w'})$ by
  construction, but there is no way to instantiate $\pvar$ and
  $\qvar$:
  since they have to agree with $\uvar$ and $\vvar$ on the $i$-th
  input position, the only candidates are $v$ and $v'$, but
  $\OUT{v} = \OUT{v'}$ by construction.
  \qed
\end{proof}

Perhaps surprisingly, $\dmHyp$ is semantically \emph{white-box}
monitorable for violations \Pty{wb-mon}.
That is, if $f$ is not DDM, there is hope to detect it.
To make this statement more precise, we first need to identify the set
of valid system behaviors $\Sys[f]$ of $f$.
We define $\ObsEv[f] = \setof{(x, y)}{f(x) = y}$ to be the set of I/O
pairs that correspond to executions of $f$.
Then $\Sys[f] = \powerset(\ObsEv[f])$ precisely characterizes the set of
valid system behaviors.

\begin{example}
  Define $g \colon \NN \times \NN \to \NN$ as $g(x, y) = x$, \ie $g$
  simply ignores its second argument.
  Then $\ObsEv[g] = \setof{ ((x, y), x) }{ x, y \in \NN }$.
  It is easy to show that DDM is white-box monitorable for $g$.
  Any finite set of valid traces $U$ can be extended to include a pair
  of traces $u, u'$ that only differ in their second input value, \eg
  $u = ((1, 1), 1)$ and $u' = ((1, 2), 1)$.
  Now, consider any $T \in \Sys[f]$ that extends
  $U \cup \set{ u, u' }$.
  Clearly, $T$ cannot contain any trace $v$ for which
  $\proj_1(\IN{v}) = 1$ but $\OUT{v} \neq 1$ as that would constitute
  an invalid system behavior.
  But $T$ would have to contain such a trace to be a model of
  $\dmHypSub{2}$.
  Hence, $T \not\models \dmHyp$ for any such $T$, which means
  $U \cup \set{ u, u' }$ permanently violates $\dmHyp$.
\end{example}
Note the crucial use of information about $g$ in the above example:
it is the restriction to \emph{valid} extensions $T \in \Sys[f]$ that
excludes trivial models such as $\Ev[f]$ and thereby restores
(semantic) monitorability for violations.
The apparent conflict between \Pty{bb-nonmon} and \Pty{wb-mon} is thus
resolved.

With the extra information that gray-box
monitoring affords, we can make more precise claims about properties
like DDM:
whether or not a property is monitorable may, for instance, depend on
whether the property actually holds for the system under scrutiny.
Concretely, for the case of DDM, we show the following.

\newcounter{thm-weak-mon}
\setcounter{thm-weak-mon}{\value{theorem}}

\begin{theorem}\label{thm:weak-mon}
  Given a function $f\colon I \to O$, the formula $\dmHyp$ is
  semantically gray-box monitorable in $\Sys[f]$ if and only if either
  $f$ is distributed non-minimal or the input domain $I$ is finite.
\end{theorem}
\noindent \Thm{weak-mon} follows from the following two auxiliary
lemmas.
\begin{lemma}[semantic violation]\label{lem:weak-viol}
  If $f$ is not DDM, then $\dmHyp$ is semantically monitorable for
  violation \textup{(}in $\Sys[f]$\textup{)}.
\end{lemma}
\begin{proof}
  Assume a finite set of traces $U \in \Sys[f]$.
  We need to show that there is a finite extension $V \extends U$
  permitted by $\Sys[f]$ that permanently violates $\dmHyp$.
  First, note that the task is trivial if $I$ is finite:
  we simply pick $V = \ObsEv[f]$, \ie the set of all possible
  executions, which is also finite.
  The only finite extension of $V$ permitted by $\Sys[f]$ is the
  complete set of traces $\ObsEv[f]$ itself, and since $f$ is not
  distributed minimal, $\dmHyp$ cannot hold for $\ObsEv[f]$.

  Assume instead that $I$ is infinite.
  Since $f$ is distributed non-minimal, there must be some input
  position $i$ and some pair of distinct inputs $x \neq x' \in I_i$,
  such that $f(z[i \mapsto x]) = f(z[i \mapsto x'])$ for any choice of
  $z \in I$.
  Let $y = z[i \mapsto x]$ and $y' = z[i \mapsto x']$ for an arbitrary
  $z \in I$.
  Then any set $W \in \Sys[f]$ that contains the traces
  $u = (y, f(y))$ and $u' = (y', f(y'))$ violates $\dmHyp$.
  To see this, assume instead that $W \sat[{\Sys[f]}] \dmHyp$.
  Then there must be traces $v, v' \in W$ that agree on all but the
  $i$-th input, such that
  $f(\IN{v}[i \mapsto x]) \neq f(\IN{v'}[i \mapsto x'])$, thus
  contradicting non-minimality of $f$.
  Hence, by picking $V = U \cup \set{u, u'}$, we have
  $V \viol[f] \dmHyp$.
  \qed
\end{proof}

\begin{lemma}[Semantic satisfaction]\label{lem:weak-sat}
  If $f\colon I \to O$ is DDM, then $\dmHyp$ is semantic monitorable for
  satisfaction \textup{(}in $\Sys[f]$\textup{)} if and only if $I$ is
  finite.
\end{lemma}
\begin{proof}
  First, if $I$ is finite the result follows by picking
  $V = \ObsEv[f]$.
  Assume now that $f$ is distributed minimal, $\dmHyp$ is semantically
  monitorable for satisfaction, and $I$ is infinite.
  Let $U \in \Sys[f]$ be some non-empty, finite set of traces with
  some distinguished element $u \in U$.
  Since $\dmHyp$ is monitorable for satisfaction, there must be a
  finite extension $V \extends U$ that permanently satisfies $\dmHyp$.
  To arrive at a contradiction, it suffices to construct a finite
  extension $W \extends V$ that does not satisfy $\dmHyp$.

  Pick an input position $i$ for which $I_i$ is infinite.
  Such an~$i$ must exist because otherwise $I$ would be the Cartesian
  product of finite sets, and $I$ is infinite by assumption.
  Next, pick a pair of distinct element $x \neq x' \in I_i$ such that
  there are no traces in $V$ with $x$ or $x'$ as their $i$-th input.
  Such $x, x'$ must also exist because $I_i$ is infinite but $V$ is
  finite.
  Finally, pick an input position $j \neq i$, and a $y \in I_j$ such
  that $y \neq \proj_j(\IN{u})$.
  Such a $y$ must exist for $I_j$ to be non-trivial.

  Now let $z = \IN{u}[i \mapsto x]$,
  $z' = \IN{u}[i \mapsto z', j \mapsto y]$ and $w = (z, f(z))$,
  $w' = (z', f(z'))$.
  Then $w$ and $w'$ are clearly valid traces, \ie
  $w, w' \in \ObsEv[f]$, but $w, w' \notin V$ since $w$ and $w'$ have
  $x$ and $x'$ as their $i$-th inputs, respectively.
  Let $W = V \cup \set{w, w'}$.
  By construction, $\neg\operatorname{same}_i(\uvar, \vvar)$ holds if
  we instantiate $\uvar$ and $\vvar$ to $w$ and $w'$, respectively,
  but there is no pair of traces $v, v' \in W$ to instantiate
  $\pvar, \qvar$ in such a way that
  $\operatorname{same}_i(\uvar, \pvar)$,
  $\operatorname{same}_i(\vvar, \qvar)$ and
  $\operatorname{almost}_i(\pvar, \qvar)$ all hold simultaneously.
  The former force the choice $\pvar \mapsto w$ and $\qvar \mapsto w'$
  but, by construction, $\proj_j(\IN{w}) \neq \proj_j(\IN{w'})$.
  Hence $W \nsat \dmHyp$ and we arrive at a contradiction.
  \qed
\end{proof}

Intuitively, Theorem~\ref{thm:weak-mon} means that $f$ cannot be
monitored for satisfaction.
Note that the semantic monitorability property established by
\Thm{weak-mon} is independent of whether we can actually decide DDM
for the given $f$.
We address the question of strong monitorability later on in this
section.

If $I$ is finite, it is easy to strengthen \Thm{weak-mon} by providing
a perfect monitor $M_{\textsf{dm}}$ for $\dmHyp$.
Since $f$ is assumed to be a total function with a finite domain, we
can simply check the validity of $\dmHyp$ for every trace
$U \subseteq \ObsEv[f]$ and tabulate the result.
To do so, the $\exists$ and $\forall$ quantifiers in $\dmHyp$ can be
converted into conjunctions and disjunctions over $U$.
\begin{corollary}\label{cor:nonmon}
  For $f:I\Into O$ with finite $I$, $\dmHyp$ is strongly monitorable in
  $\Sys[f]$.
\end{corollary}
If $I$ is infinite, then $\dmHyp$ is not semantically monitorable for
satisfaction, but we can still hope to build a sound monitor for
violation of $\dmHyp$.

\subsection{Building a Gray-box Monitor for DDM}

In what follows, we assume a computable function capable of deciding
DDM only for some instances.
This function, which we call oracle, will serve as the basis for a
\emph{sound} monitor for DDM \Pty{smon}.
This monitor will detect some, but not all, violations of DDM when
given sets of observed traces, thus resolving the apparent tension
between~\Pty{undec} and \Pty{smon}.

Given $f\colon I_1 \times \dotsm \times I_n \to O$, we define the
predicate $\phi_f$ as
\[
  \phi_f(i, x, y) = \exists z \in I.\, f(z[i \mapsto x]) \neq f(z[i
  \mapsto y]),
\]
and assume a total computable function
$N_{f, i}\colon I_i \times I_i \to \set{\top, \bot, ?}$ such that
\[
N_{f,i}(x,y) = \begin{cases}
  \top \text{ or } ? & \text{if $\phi_f(i, x, y)$ holds},\\
  \bot \text{ or } ? & \text{otherwise.}
\end{cases}
\]
The function $N_{f,i}$ acts as our oracle to instantiate the
existential quantifiers in $\dmHyp$.
As discussed earlier, such oracles may be implemented by statically
analyzing the system under observation (here, the function $f$).
In our proof-of-concept implementation, we extract $\phi_f(i, x, y)$
from $f$ using {\em symbolic execution}, and use an SMT solver to
compute $N_{f, i}(x, y)$.

We now define a monitor $M_{\textsf{dm}}$ for $\dmHyp$ as follows:
\[
  M_{\textsf{dm}}(U) = \begin{cases}
    ?{} & \text{if } f(\IN{u}) \neq \OUT{u} \text{ for some } u \in U,\\
    ?{} & \text{if } \bigand_{i=1}^n \bigand_{u,u' \in U}
    N_{f,i}(\proj_i(\IN{u}), \proj_i(\IN{u'})) \neq \bot,\\
    \bot & \text{otherwise.}
  \end{cases}
\]
Intuitively, the monitor $M_{\textsf{dm}}(U)$ checks the set of traces $U$
for violations of DDM by verifying two conditions:
the first condition ensures the \emph{consistency} of $U$, \ie that
every trace in $U$ does in fact correspond to a valid execution of
$f$;
the second condition is necessary for $U$ \emph{not} to permanently
violate $\dmHyp$.
Hence, if it fails, $U$ must permanently violate $\dmHyp$.
Since $N_{f,i}$ is computable, so is $M_{\textsf{dm}}$.
Note that $M_{\textsf{dm}}$ never gives a positive verdict $\top$.
This is a consequence of \Thm{weak-mon}:
if $f$ is DDM, then $\dmHyp$ is not monitorable in $\Sys[f]$.
In other words, DDM is not monitorable for satisfaction.

The second condition in the definition of $M_{\textsf{dm}}$ is an
approximation of $\dmHyp$:
the universal quantifiers are replaced by conjunctions over the finite
set of input traces $U$, while the existential quantifiers are
replaced by a single quantifier ranging over all of $\ObsEv[f]$ (not
just $U$).
This approximation is justified formally by the following theorem.

\newcounter{thm-soundness}
\setcounter{thm-soundness}{\value{theorem}}

\begin{theorem}[soundness]\label{thm:soundness}
  The monitor $M_{\textsf{dm}}$ is sound. Formally,
  \begin{enumerate}
  \item $U \sat[{\Sys[f]}] \dmHyp$ if $M_{\textsf{dm}}(U) = \top$, and
  \item $U \viol[{\Sys[f]}] \dmHyp$ if $M_{\textsf{dm}}(U) = \bot$.
  \end{enumerate}
\end{theorem}

\begin{proof}
  The monitor never gives a $\top$ verdict, so the first half of the
  theorem (satisfaction) holds vacuously.
  For the second part (violation), we have
  \begin{align*}
    M_{\textsf{dm}}(U) = \bot
    &\textstyle \quad \iff \quad U \in \Sys[f] \, \land \,
      \bigor_{i=1}^n \bigor_{u,u' \in U} N_{f,i}(\proj_i(\IN{u}), \proj_i(\IN{u'})) = \bot,
  \end{align*}
  and
  \begin{align*}
    &\textstyle
      \bigor_{i=1}^n \bigor_{u, u' \in U}
      N_{f,i}(\proj_i(\IN{u}), \proj_i(\IN{u'})) = \bot\\
    \iff\quad
    &\textstyle
      \bigor_i \exists u, u' \in U.\,
      \neg\phi_f(i, \proj_i(\IN{u}), \proj_i(\IN{u'}))\\
    \iff\quad
    &\textstyle
      \bigor_i \exists u, u' \in U.\, \forall z \in I.\,
      f(z[i \mapsto \proj_i(\IN{u})]) = f(z[i \mapsto \proj_i(\IN{u'})])\\
    \iff\quad
    &\textstyle
      \bigor_i \exists u, u' \in U.\, \forall w \in \ObsEv[f].\,
      f(\IN{w}[i \mapsto \proj_i(\IN{u})]) = f(\IN{w}[i \mapsto \proj_i(\IN{u'})])\\
    \implies\quad
    & \forall V \in \Sys[f]. \, U \prefix V \implies\\
    &\textstyle \quad
      \bigor_i \exists u, u' \in V.\, \forall w \in V.\,
      f(\IN{w}[i \mapsto \proj_i(u)]) = f(\IN{w}[i \mapsto \proj_i(u')])\\
    \iff \quad &U \viol[{\Sys[f]}] \dmHyp.
  \end{align*}
  \qed
\end{proof}

We describe a prototype implementation of $M_{\textsf{dm}}$
in~\Sec{minion}.

\enableLstShortInline

\section{Implementation and Prototype}
\label{app:minion}
\label{sec:minion}

We have implemented the ideas described in \Sec{monitoring} in a
proof-of-concept monitor for DDM called \Minion{}.
The monitor uses the symbolic execution API and the SMT backend of the
KeY deductive verification system~\cite{ahrendt16deductive,KeYproject}
to extract logical characterizations of Java programs (their {\em
  symbolic execution trees}).
It then extends them to first-order formulas over sets of observed
traces, and checks the result using the state-of-the-art SMT solver
Z3~\cite{DeMouraB08tacas,Z3github}.
The \Minion{} monitor is written in Scala and provides a simple
command-line interface (CLI).
Its source code is freely available online at
\url{https://github.com/sstucki/minion/}.

\begin{figure}%
  \begin{lstlisting}
 class Toll {
   int rate(int hour, int passengers) {
     int r;                                     // standard rates:
     if (hour >= 9 && hour <= 17) { r = 90; }   //  - daytime
     else                         { r = 70; }   //  - nighttime
     if (passengers > 2) { r = r - (r / 5); }   // carpool: 20%
     return r;
   }
   int max(int x, int y) {
     if (x > y) { return x; } else { return y; }
   }
   int fee(int t1, int t2, int t3, int p) {
     int r1 = rate(t1, p);           // rates at each toll station
     int r2 = rate(t2, p);
     int r3 = rate(t3, p);
     int f1 = max(r1, r2) * 4;       // fees per road section
     int f2 = max(r2, r3) * 7;
     return f1 + f2;                 // total fee
   }
 }
  \end{lstlisting}\vspace{-1em}
  \caption[Toll fees]{A program for computing the total fee of a trip
    on a toll road.}
  \label{fig:toll}
\end{figure}

Before we describe \Minion{} in more detail, we introduce a running
example illustrating the principles of both monolithic and distributed
data minimality.
For an example of monolithic data minimization, first consider the
method @rate@ shown in \Fig{toll}.
The purpose of this method is to compute the baseline rate to be paid
by the driver of a vehicle on a toll road.
The rate depends on the time of day and the number of passengers in
the vehicle.
The range of the output is $\{56, 70, 72, 90 \}$, and consequently the
data processor does not need to know the precise hour of the day, nor
the exact number of passengers.
A vehicle might pass a toll station at any time between 9pm and 5am
to be subject to a the higher daytime rates (72, 90), and at any other
time to benefit from the lower nighttime rates (56, 70).
Also, any vehicle occupied by three or more passengers is eligible for
a 20\% carpool discount.
Giving the actual hour and number of passengers violates the principle
of data minimality because more information than necessary is
collected.
Data minimization is the process of ensuring that the range of inputs
provided is reduced, such that different inputs result in different
outputs.

In a distributed setting, the concept of minimization is more complex
as input data may be collected from multiple independent sources.
Consider the method @fee@ in \Fig{toll}.
This method computes the total fee for a trip on a toll road, based on
the hours at which a vehicle passes three consecutive toll stations,
and on the number of passengers in the vehicle.
The overall fee depends on the total time spent on the toll road,
which is data collected from all three toll stations.
In particular, if a vehicle enters a section of the toll road during a
low-rate early morning hour, but fails to reach the next station
before 9pm, the driver will be charged the more expensive daytime rate
for the entire section.
\DDM requires minimizing each input parameter (\ie the information
collected at separate toll stations) individually.
A \emph{preprocessor} or \emph{data minimizer}~\cite{ASS17dm} located
at any given toll station can easily minimize the individual inputs
(@hour@, @passengers@) at that station.
But an individual minimizer cannot guarantee minimization with
respect to the overall fee since it has no information about the input
data collected at the other stations.
\DDM therefore constitutes merely a ``best effort'' to minimize inputs
given the inherently distributed nature of the system.

When running \Minion{} on the @fee@ method of the class~@Toll@, the
tool builds first the symbolic execution tree.
Then, the monitor reads and parses traces from an input file or
standard input.
Whenever \Minion{} parses a new trace, it rechecks the entire set of
traces read thus far for violation, thereby supporting both online and
offline monitoring.
Traces are read from CSV files, where the number and format of the
inputs is determined automatically from the method signature.
\Fig{raw-traces} shows example traces for the @fee@ method.
Columns~1--4 correspond to the parameters @h1@, @h2@, @h3@ and @p@,
respectively, while column~5 contains the result computed by @fee@ for
the given values.

By default, \Minion{} monitors traces for DDM.
Thus, when processing the traces given in \Fig{raw-traces}, it signals
a violation after reading the second line because @fee@$(20,h_2,h_3,p)
={}$@fee@$(2, h_2, h_3, p)$ irrespective of the choice of $h_2$, $h_3$,
and $p$.
In contrast, all traces listed in \Fig{dist-traces} are accepted by
\Minion{} since they have been preprocessed by a distributed
minimizer.
Alternatively, \Minion{} can be instructed to monitor traces for
monolithic data minimality (MDM) in which case a violation is signaled
when processing the last line of \Fig{dist-traces}, whereas all traces
in \Fig{mono-traces} are accepted.

\begin{figure}[tbp]
  \centering
  \begin{sfloat}{unprocessed}{fig:raw-traces}{.34\textwidth}
    \centering
\begin{verbatim}
    20, 22, 1,  1, 770
    2 , 2,  3,  5, 616
    9 , 10, 10, 4, 792
    23, 0,  2,  5, 616
    10, 11, 14, 1, 990
    8,  10, 11, 1, 990
    ...
\end{verbatim}%
  \end{sfloat}\quad%
  \begin{sfloat}{distributed~minimal}{fig:dist-traces}{.30\textwidth}
    \centering
\begin{verbatim}
    0, 0, 0, 1, 770
    0, 0, 0, 3, 616
    9, 9, 9, 3, 792
    0, 0, 0, 3, 616
    9, 9, 9, 1, 990
    0, 9, 9, 1, 990
    ...
\end{verbatim}%
  \end{sfloat}\quad%
  \begin{sfloat}{monolithic~minimal}{fig:mono-traces}{.30\textwidth}
    \centering
\begin{verbatim}
    20, 22, 1,  1, 770
    2,  2,  3,  5, 616
    9,  10, 10, 4, 792
    2,  2,  3,  5, 616
    14, 15, 15, 2, 990
    14, 15, 15, 2, 990
    ...
\end{verbatim}%
  \end{sfloat}
  \caption[Traces]{Raw and minimized traces generated from \texttt{Toll.java}.}
  \label{fig:traces}
\end{figure}

\subsection{Lazy \vs Eager Monitoring}\label{sec:eager}
Perhaps surprisingly, there are cases where \Minion{} will detect a
violation of DDM
whereas it will not detect a
violation of MDM.
Consider the function
$f(x, y) = x$.
Since $f$ simply ignores its second argument, it is clearly neither
distributed nor monolithic minimal.
When monitoring the pair of traces $(1,2,1)$ and $(3,4,3)$ for DDM,
\Minion{} detects a violation because $f(x, 2) = f(x, 4)$ for any
choice of $x$.
Note, however, that this situation does not appear among the observed
traces since the two values for $y$ in the respective traces differ.
The tool reports a violation because a common value for $x$ is found
by our oracle when monitoring for DDM.
When monitoring for MDM \Minion{} does not detect the
violation, because in this case there is no need to invoke the oracle.

Whether or not this is the intended behavior of the monitor depends on
the assumption of whether the traces are collected from a program $f$
or from the \emph{combined} program $f \comp p$ ($p$ being a minimizer).
In the latter case, some combinations of inputs may never be observed
as the inputs have been minimized.
On the other hand, if traces are not considered preprocessed, we may
wish to explore the behavior of $f$ more exhaustively.
For this purpose, \Minion{} can be instructed to monitor a set of
traces \emph{eagerly} for MDM,
\resp \emph{lazily} for DDM.
For the former, \Minion{} considers not
just the observed traces, but any combination of observed input
values---even if that combination does not actually correspond to an
observed trace.
For the latter, \Minion{} only considers combinations of inputs
originating from traces with the same result value.
For example, for the pair of input traces $(1,2,1)$ and $(3,4,3)$,
\Minion{} is able to find a violation in eager MDM mode
since $f(1,2) = f(1,4)$, but not in lazy DDM mode since
$f(1,2) \neq f(3,4)$.

\begin{figure}[tbp]
  \begin{lstlisting}
class Div {
  //@ requires x >= 0 && y > 0;
  //@ ensures (\result * y <= x) && (\result * y < x + y);
  int posDiv(int x, int y) {
    int q = 0;
    //@ maintaining (r >= 0) && (r + q * y == x);
    //@ decreasing r;
    for (int r = x; r >= y; ++q) { r -= y; }
    return q;
  }
}\end{lstlisting}\vspace{-1em}
  \caption[Division]{A naive division algorithm for positive integers.}
  \label{fig:posDiv}
\end{figure}

\subsection{Loops and loop Invariants}
Thus far, we have only considered simple programs
whose control flow does not include loops or recursive calls.
Monitoring programs with loops for data minimality is more
challenging, as illustrated by method @posDiv@ given in \Fig{posDiv},
which implements integer division by repeatedly subtracting @y@ from
@x@ and counting how many times this is possible.
Our simple symbolic tree method cannot extract a complete logical
characterization of @posDiv@ automatically.
Instead, we propose two options:
\begin{enumerate}%
\item to obtain an approximate characterization by unrolling the loop to
  a fixed depth;
\item to annotate loop invariants, which allows KeY's symbolic
  execution engine to merge the different execution paths and give a
  complete characterization of the method.
\end{enumerate}
Both of these options are implemented in \Minion{}.
In the first case, choosing large values for unrolling leads to high
symbolic execution times but lower values may affect accuracy.
If the number of loop iterations $n$ at runtime remains below the
number of unrolls $m$, the resulting logical characterization is exact
but if $n > m$ the characterization becomes an over-approximation and
the the monitor may fail to detect non-minimal traces\footnote{Indeed,
  it may not even detect inconsistent traces, \ie traces where the
  expected output differs from the actual output computed at
  runtime.}.

These false negatives can be avoided by annotating loops with JML loop
invariants, as shown in \Fig{posDiv}.
This invariant specifies that at every loop iteration the remainder
$r$ is positive and, when added to the iteration counter $q$ times the
divisor $y$, equals the original dividend $x$.
With this invariant, the symbolic execution terminates quickly and
without cutting off any branches, and \Minion{} is able to extract a
logical characterization for @posDiv@, which asserts that the eventual
result $q$ of the method must satisfy the equation $qy = x - r$ for
some $r$ such that $0 \leq r < y$.
This is sufficient to correctly monitor any traces generated from
@posDiv@ for violation of DDM.

\disableLstShortInline

\subsection{Performance evaluation}

\addtolength{\tabcolsep}{.5pt}
\begin{table}[tbp]
  \centering
  \footnotesize
  \begin{tabular}{%
    @{}lr@{\quad}rrcc@{\quad}rrcc@{~}rrcc@{}} \toprule
    & \multicolumn{1}{c}{Symb.}
    & \multicolumn{4}{c}{K1: random (s)}
    & \multicolumn{4}{c}{K2: DDMin (s)}
    & \multicolumn{4}{c}{K3: MDMin (s)} \\ %
    & \multicolumn{1}{c}{exec. (s)} & \Eager & \Lazy & E & L
    & \Eager & \Lazy & E & L & \Eager & \Lazy & E & L \\ \midrule
    T1%
    & 30.9\PM{1.5}
    &  0.6\PM{0.2}     &  0.6\PM{0.1} & \FF & \FF
    & 51.2\PM{1.2}     & 30.3\PM{7.2} & \QM & \QM
    &  0.9\PM{\Z{}0.8} &  5.4\PM{0.2} & \FF & \QM
    \\
    T2%
    &  9.2\PM{0.7}
    &  0.4\PM{0.1}     &  0.4\PM{0.0} & \FF & \FF
    & 17.3\PM{0.8}     & 13.9\PM{2.2} & \QM & \QM
    & 17.3\PM{\Z{}0.7} &  3.7\PM{0.3} & \QM & \QM
    \\
    CA%
    &  1.8\PM{0.1}
    &  0.5\PM{0.2}     &  0.4\PM{0.2} & \FF    & \FF
    & 19.2\PM{2.0}     & 14.8\PM{1.7} & \QM    & \QM
    & 13.6\PM{\Z{}5.4} &  3.6\PM{0.1} & \FF\QM & \QM
    \\
    LA%
    &    3.4\PM{0.4}
    &    0.4\PM{0.2} & 0.3\PM{0.0} & \FF & \FF
    &                &            &     &
    & 136.6\PM{37.1} & 3.7\PM{0.3} & \QM & \QM
    \\ \bottomrule
    \\
  \end{tabular}
  \caption[Traces]{Mean running times and verdicts of
    \Minion{} monitoring various methods.
  }
  \label{tbl:benchmarks}
\end{table}

\enableLstShortInline

We evaluated the performance of \Minion{} on a MacBook Pro with a
3.1~GHz Intel~Core~i5 processor and 16~GB of memory, running macOS
10.14.
The results are summarized in \Tab{benchmarks}.
We ran \Minion{} on four Java methods: the @fee@ method from
\Fig{toll} (T1), a variant of that method that computes the fee on a
road with only two toll stations instead of three (T2), as well as the
\texttt{CreditApp} (CA) and \texttt{LoyaltyApp} (LA) benchmarks
introduced in~\cite{corrASS16dm}.
Each method was monitored for DDM violation using three kinds of input
traces:
\begin{enumerate}[(K1)]
\item random input values that respect the input specifications of the
  methods;
\item traces from (K1) minimized using a \emph{distributed data
    minimizer} (DDMin);
\item traces from (K1) minimized using a \emph{monolithic data
    minimizer} (MDMin).
\end{enumerate}
The traces shown in \Fig{traces} are subsets of the inputs generated
for T1.
We generated 10~instances of each kind, accounting for a total of
30~trace sets, each containing exactly 100~traces.

\Tab{benchmarks} shows the mean running time and standard deviation in
seconds, as well as the verdicts produced by \Minion{}.
The second column of the table reports the time spent by the symbolic
execution.
T1 incurs in higher running times because T1 features several
multiply-nested branches.
The remaining columns report the execution times and verdict of the
actual monitor.

The performance of eager and lazy monitoring is similar on random (K1)
inputs because all cases have (finite) small input and output domains.
As expected, the verdicts for the DDMin (K2) traces are inconclusive
since DDM is not monitorable for satisfiability in general (by
\Lem{weak-sat}).
Lazy monitoring does consistently better than eager monitoring on
DDMin (K2) inputs, though the differences are relatively small for T1,
T2 and CA because the ranges of these methods are small
($<10$~elements).
There is a bigger difference for LA, where the range is larger.

The performance of lazy monitoring on MDMin traces (K1) is
consistently better than on DDMin traces (K2) because lazy DDM
monitoring and MDM monitoring coincide for MDMin traces (no SMT
invocations are necessary).
On the other hand, the performance of eager monitoring for MDMin
traces may change drastically depending on whether or not the traces
are also DDMin (which need not be the case).
If they are, then eager monitoring has the same performance for MDMin
traces as for DDMin traces.
If they are not, the eager monitor might detect a violation early in
the input set, cutting the overall execution time.

\disableLstShortInline
\section{Related Work}
\label{sec:related}

\myparagraph{LTL Monitorability.}
Pnueli and Zaks~\cite{pnueli06psl} introduced monitorability as the
existence of extension of the observed traces that permanently satisfy
or violate an LTL property.
It is known that the set of monitorable LTL properties is a superset
of the union of safety and co-safety
properties~\cite{bauer11runtime,bauer07good} and that it is also a
superset of the set of obligation properties~\cite{falcone09runtime,
  falcone12what}.
Havelund and Peled~\cite{havelund18runtime} introduce a finer-grained
taxonomy distinguishing between \emph{always} finitely satisfiable
(resp. refutable), and \emph{sometimes} finitely satisfiable where only
some prefixes are required to be monitorable (for satisfaction).
Their taxonomy also describes the relation between monitorability and
classical safety properties.
This is a new dimension in the monitorability cube in
\Fig{cube} which we will study in the future.
While all the notions mentioned above ignore the system, predictive
monitoring~\cite{zhang12runtime} considers the traces allowed in a
given finite state system.

\myparagraph{Monitoring HyperLTL.}
Monitoring hyperproperties was first studied in~\cite{agrawal16runtime},
which introduces the notion of monitorability for
HyperLTL~\cite{cfkmrs14} and gives an algorithm for a
fragment of alternation-free HyperLTL.
This is later generalized to the full fragment of alternation-free
formulas using formula rewriting in~\cite{bsb17}, which can also
monitor alternating formulas but only with respect to a fixed finite
set of finite traces.
Finally, \cite{finkbeiner17monitoring} proposes an automata-based
algorithm for monitoring HyperLTL, which also produces a monitoring
verdict for alternating formulas, but again for a fixed trace set.
The complexity of monitoring different fragments of HyperLTL was
studied in detail in~\cite{bf18}.
The idea of gray-box monitoring for hyperproperties, as a means for
handling non-monitoriable formulas, was first proposed
in~\cite{bss18}.

\myparagraph{Data minimization.}
A formal definition of {\it data minimization} and the concept of {\it
  data minimizer} as a preprocessor appear in~\cite{ASS17dm}, which
introduces the monolithic and distributed cases.
Minimality is closely related to information flow \cite{cohen1977}.
Malacaria et al.~\cite{malacaria2016} present a symbolic
execution-based verification of non-interference security properties
for the \emph{OpenSSL} library.
In our paper, we have focused on a version of distributed minimization
which is not monitorable in general.
For stronger versions (cf.~\cite{ASS17dm}), Pinisetty et
al.~\cite{PSS18rvh,PASS18corr} show that monitorability for
satisfaction is not possible, but it is for violation.
(the paper also introduces an RV approach for similar safety
hyperproperties for deterministic programs).

\section{Conclusions}
\label{sec:conclusions}

We have rephrased the notion of monitorability considering different
dimensions, namely 
\begin{inparaenum}[(1)]
\item whether the monitoring is black-box or gray-box,
\item whether we consider trace properties or hyperproperties, and
\item taking into account the computatibility aspects of the monitor
  as a program.
\end{inparaenum}
We showed that many hyperproperties that involve quantifier
alternation are non-monitorable in a black-box manner and proposed a
technique that involves inspecting the behavior of the system.
In turn, this forces to consider the computability limitations of the
monitor, which leads to a more general notion of monitorability.

We have considered distributed data minimality (DDM) and expressed
this property in HyperLTL, involving one quantifier alternation.
We then presented a methodology to monitor violations of DDM, based on
a model extracted from the program being monitored in the form of its
symbolic execution tree, and an SMT solver.
We have implemented a tool (\Minion{}) and applied it to a number of
representative examples to assess the feasibility of our approach.

As future work, we plan to extend the proposed methodology for other
hyperproperties, particularly in the concurrent and distributed
setting. We are also planning to use bounded model checking as our verifier at 
run-time by combining over- and under-approximated methods to deal with 
universal and existential quantifiers in HyperLTL formulas. Another interesting 
problem is to apply gray-box monitoring for hyperproperties with real-valued 
signals (e.g., HyperSTL~\cite{nkjdj17}).
Finally, we intend to extend the definition and results of data
minimality in order to capture reactivity, and study monitorability in
this setting.

\subsubsection*{Acknowledgements}
This research has been partially supported
by the United States NSF SaTC Award~1813388,
by the Swedish Research Council ({\it Vetenskapsr\aa det}) under
Grant~2015-04154 ``PolUser'',
by the Madrid Regional Government under Project~S2018/TCS-4339
``BLOQUES-CM'',
by EU H2020 Project~731535 ``Elastest'', and
by Spanish National Project~PGC2018-102210-B-100 ``BOSCO''.

\newpage

\renewcommand{\bibsection}{\section*{References}}

\bibliographystyle{tex-includes/splncsnat}
\begingroup
  \microtypecontext{expansion=sloppy}
  \small %
  \bibliography{paper}
\endgroup

\vfill

\end{document}